\documentclass[fleqn,10pt]{SelfArx}

\usepackage[english]{babel}
\usepackage{amsmath,amssymb,mathtools}
\usepackage{amsthm}
\usepackage{array}
\usepackage{booktabs}
\usepackage{listings}
\usepackage{xcolor}
\usepackage{tikz}
\usetikzlibrary{arrows.meta,decorations.pathreplacing}
\usepackage{seqsplit}
\usepackage{orcidlink}
\usepackage{csquotes}
\usepackage{enumitem}
\usepackage[style=numeric,sorting=none,backend=biber]{biblatex}
\usepackage{hyperref}

\theoremstyle{plain}
\newtheorem{theorem}{Theorem}
\newtheorem{lemma}[theorem]{Lemma}
\newtheorem{proposition}[theorem]{Proposition}
\theoremstyle{definition}
\newtheorem{definition}[theorem]{Definition}

\definecolor{color1}{RGB}{0,0,90}
\definecolor{color2}{RGB}{0,20,20}
\definecolor{MissingPairCell}{RGB}{214,70,70}
\definecolor{PresentPairCell}{RGB}{238,242,247}
\definecolor{MatrixRule}{RGB}{35,35,35}
\definecolor{FlowArrow}{RGB}{75,92,112}
\definecolor{FiberExtra}{RGB}{74,119,168}
\definecolor{FiberText}{RGB}{44,55,68}

\newcommand{\Lean}{Lean}
\newcommand{\mathlib}{\textsf{\small mathlib}}
\newcommand{\Fin}{\operatorname{Fin}}
\newcommand{\N}{\mathbb{N}}
\newcommand{\distH}{d_h}

\newcommand{\ArtifactRepoUrl}{https://github.com/florath/covering-codes-lean}
\newcommand{\ArtifactCommit}{fae817df3d5b8562c8fa27bc009050a797ddb5b6}

\newcommand{\ArtifactBlob}{\ArtifactRepoUrl/blob/\ArtifactCommit}
\newcommand{\artifactrepo}{\href{\ArtifactRepoUrl}{\path{github.com/florath/covering-codes-lean}}}
\newcommand{\artifactcommit}{\href{\ArtifactRepoUrl/tree/\ArtifactCommit}{\texttt{\ArtifactCommit}}}
\newcommand{\leanfile}[1]{\href{\ArtifactBlob/#1}{\path{#1}}}
\newcommand{\leanfileat}[2]{\href{\ArtifactBlob/#1\#L#2}{\path{#1}}}
\newcommand{\leandir}[1]{\href{\ArtifactRepoUrl/tree/\ArtifactCommit/#1}{\path{#1}}}
\newcommand{\leansource}[3]{\href{\ArtifactBlob/#2\#L#3}{#1}}

\newcommand{\leanfilelabel}[2]{\href{\ArtifactBlob/#2}{\texttt{#1}}}
\newcommand{\leandirlabel}[2]{\href{\ArtifactRepoUrl/tree/\ArtifactCommit/#2}{\texttt{#1}}}
\newcommand{\leanmain}[2]{\leansource{#1}{CoveringCodes/Database/Sources/OctonaryFourTwo.lean}{#2}}
\newcommand{\leansparse}[2]{\leansource{#1}{CoveringCodes/Database/Sources/SparseSlicer.lean}{#2}}
\newcommand{\leanlrat}[2]{\leansource{#1}{CoveringCodes/Database/Sources/OctonaryFourTwoBlockLRAT.lean}{#2}}
\newcommand{\leancovernum}[2]{\leansource{#1}{CoveringCodes/CoveringNumber.lean}{#2}}
\newcommand{\leanexplicit}[2]{\leansource{#1}{CoveringCodes/Database/ExplicitCode.lean}{#2}}
\makeatletter
\newcommand{\leandecl}[3]{%
  \href{\ArtifactBlob/#2\#L#3}{%
    \begingroup\ttfamily
    \edef\@tempa{\detokenize{#1}}%
    \expandafter\seqsplit\expandafter{\@tempa}%
    \endgroup}}
\makeatother
\newcommand{\leanmaindecl}[2]{\leandecl{#1}{CoveringCodes/Database/Sources/OctonaryFourTwo.lean}{#2}}
\newcommand{\leansparsedecl}[2]{\leandecl{#1}{CoveringCodes/Database/Sources/SparseSlicer.lean}{#2}}
\newcommand{\leanlratdecl}[2]{\leandecl{#1}{CoveringCodes/Database/Sources/OctonaryFourTwoBlockLRAT.lean}{#2}}
\newcommand{\LeanRefs}[1]{%
  \par\smallskip\noindent
  {\footnotesize\raggedright\textit{Lean implementation.} #1\par}%
  \smallskip}

\newcommand{\FiberCountDiagram}{%
\begin{tikzpicture}[
  x=0.68cm,
  y=0.50cm,
  basecell/.style={draw=white,line width=0.3pt,fill=PresentPairCell},
  extraarea/.style={
    draw=FiberExtra,
    dashed,
    line width=0.55pt,
    fill=FiberExtra!9
  },
  label/.style={font=\scriptsize,text=FiberText},
  formula/.style={font=\small,text=FiberText,align=center},
  casebox/.style={
    rectangle,
    rounded corners=2pt,
    draw=MatrixRule,
    fill=PresentPairCell,
    line width=0.45pt,
    align=center,
    inner xsep=5pt,
    inner ysep=3pt,
    font=\small,
    text width=2.35cm,
    minimum height=0.95cm
  }
]
  \node[formula] at (4.1,6.35)
    {\(\displaystyle C^\ast=\bigsqcup_{a\in\Fin 8} C^\ast_{i,a}\)\\[2pt]
     \(\displaystyle |C^\ast_{i,a}|=2+\varepsilon_i(a)\),
     \(\varepsilon_i(a)\in\N\)};

  \foreach \a in {0,...,7} {
    \pgfmathsetmacro{\x}{1.05*\a}
    \pgfmathsetmacro{\h}{0.35 + 0.16*mod(3*\a+1,5)}
    \filldraw[basecell] (\x,0) rectangle ++(0.82,0.72);
    \filldraw[basecell] (\x,0.78) rectangle ++(0.82,0.72);
    \filldraw[extraarea] (\x,2.05) rectangle ++(0.82,\h);
    \node[label] at (\x+0.41,3.12) {\(\varepsilon_i(\a)\)};
    \node[label] at (\x+0.41,-0.36) {\(a=\a\)};
  }

  \draw[MatrixRule,line width=0.55pt] (-0.08,1.74) -- (8.22,1.74);
  \node[label,anchor=east] at (-0.25,2.51) {variable extra};
  \node[label,anchor=east] at (-0.25,1.74) {baseline};
  \node[label,anchor=east] at (-0.25,0.75) {two forced};

  \node[label] at (4.1,-0.9)
    {column \(a\) is the fiber \(C^\ast_{i,a}\)};

  \draw[decorate,decoration={brace,amplitude=4pt,mirror},MatrixRule,line width=0.45pt]
    (0,-1.22) -- (8.17,-1.22)
    node[midway,yshift=-0.43cm,formula]
      {\(8\cdot 2=16\) forced baseline};

  \draw[decorate,decoration={brace,amplitude=4pt},FiberExtra,line width=0.5pt]
    (0,3.48) -- (8.17,3.48)
    node[midway,yshift=0.36cm,formula]
      {\(\sum_{a\in\Fin 8}\varepsilon_i(a)=|C^\ast|-16\)};

  \node[casebox] at (1.25,-3.45)
    {\(|C^\ast|\le 21\)\\[1pt]
     total extra \(\le 5\)};
  \node[casebox] at (6.95,-3.45)
    {\(|C^\ast|=22\)\\[1pt]
     total extra \(=6\)};
\end{tikzpicture}%
}

\newcommand{\MissingPairBlockMatrix}{%
\begin{tikzpicture}[x=0.34cm,y=0.34cm]
  \node[anchor=south,font=\scriptsize] at (4,8.95) {symbol in coordinate \(j\)};
  \node[rotate=90,anchor=south,font=\scriptsize] at (-1.8,4) {symbol in coordinate \(i\)};

  \foreach \x in {0,...,7} {
    \node[font=\scriptsize] at (\x+0.5,8.45) {\x};
    \node[font=\scriptsize] at (-0.35,7.5-\x) {\x};
  }

  \foreach \r in {0,...,7} {
    \foreach \c in {0,...,7} {
      \pgfmathtruncatemacro{\rowblock}{ifthenelse(\r<3,1,ifthenelse(\r<6,2,3))}
      \pgfmathtruncatemacro{\colblock}{ifthenelse(\c<3,1,ifthenelse(\c<6,2,3))}
      \ifnum\rowblock=\colblock
        \fill[PresentPairCell] (\c,7-\r) rectangle ++(1,1);
      \else
        \fill[MissingPairCell!78] (\c,7-\r) rectangle ++(1,1);
      \fi
      \draw[white,line width=0.2pt] (\c,7-\r) rectangle ++(1,1);
    }
  }

  \draw[MatrixRule,line width=0.6pt] (0,0) rectangle (8,8);
  \foreach \b in {3,6} {
    \draw[MatrixRule,line width=0.75pt] (\b,0) -- (\b,8);
  }
  \foreach \b in {2,5} {
    \draw[MatrixRule,line width=0.75pt] (0,\b) -- (8,\b);
  }

  \node[font=\scriptsize,anchor=north] at (1.5,-0.18) {\(A_j^{(1)}\)};
  \node[font=\scriptsize,anchor=north] at (4.5,-0.18) {\(A_j^{(2)}\)};
  \node[font=\scriptsize,anchor=north] at (7,-0.18) {\(A_j^{(3)}\)};
  \node[font=\scriptsize,anchor=east] at (-0.65,6.5) {\(A_i^{(1)}\)};
  \node[font=\scriptsize,anchor=east] at (-0.65,3.5) {\(A_i^{(2)}\)};
  \node[font=\scriptsize,anchor=east] at (-0.65,1) {\(A_i^{(3)}\)};

\end{tikzpicture}%
}

\newcommand{\CertificateAForcedEdgeDiagram}{%
\begin{tikzpicture}[
  coordbox/.style={
    rectangle,
    draw=MatrixRule,
    fill=PresentPairCell,
    line width=0.45pt,
    rounded corners=1pt,
    align=center,
    inner xsep=4pt,
    inner ysep=4pt,
    font=\scriptsize,
    text width=1.65cm
  },
  fixedcoord/.style={
    coordbox,
    fill=MissingPairCell!9,
    draw=MissingPairCell!85
  },
  fixededge/.style={MissingPairCell,line width=1.05pt},
  pendingedge/.style={MatrixRule!35,line width=0.55pt,dashed},
  fixedlabel/.style={
    fill=white,
    inner sep=1pt,
    font=\scriptsize,
    text=MissingPairCell
  },
  pendinglabel/.style={
    fill=white,
    inner sep=1pt,
    font=\scriptsize,
    text=MatrixRule!60
  },
  formula/.style={
    rectangle,
    draw=MissingPairCell!60,
    fill=white,
    line width=0.35pt,
    rounded corners=1pt,
    align=center,
    inner xsep=4pt,
    inner ysep=3pt,
    font=\scriptsize
  }
]
  \node[fixedcoord] (c0) at (0,3.0)
    {coordinate \(0\)\\[2pt]
     \(B_0^{(1)}\sqcup B_0^{(2)}\sqcup\{6,7\}\)};
  \node[fixedcoord] (c1) at (5.4,3.0)
    {coordinate \(1\)\\[2pt]
     \(B_1^{(1)}\sqcup B_1^{(2)}\sqcup\{6,7\}\)};
  \node[coordbox] (c2) at (5.4,0)
    {coordinate \(2\)\\[2pt]
     lower profile};
  \node[coordbox] (c3) at (0,0)
    {coordinate \(3\)\\[2pt]
     lower profile};

  \draw[pendingedge] (c1) -- node[pendinglabel,right] {\(E_{12}\)} (c2);
  \draw[pendingedge] (c2) -- node[pendinglabel,below] {\(E_{23}\)} (c3);
  \draw[pendingedge] (c3) -- node[pendinglabel,left] {\(E_{03}\)} (c0);
  \draw[pendingedge] (c0) -- node[pendinglabel,sloped,above,pos=0.30] {\(E_{02}\)} (c2);
  \draw[pendingedge] (c1) -- node[pendinglabel,sloped,above,pos=0.70] {\(E_{13}\)} (c3);
  \draw[fixededge] (c0) -- node[fixedlabel,above] {\(E_{01}\)} (c1);

  \node[formula] at (2.7,1.5)
    {\(\displaystyle E_{01}=\bigcup_{r\ne s}
      B_0^{(r)}\times B_1^{(s)}\)\\[-1pt]
     only \(E_{01}\) is fixed here};
\end{tikzpicture}%
}

\newcommand{\CommonBlockSystemDiagram}{%
\begin{tikzpicture}[
  coordbox/.style={
    rectangle,
    draw=MatrixRule,
    fill=PresentPairCell,
    line width=0.45pt,
    rounded corners=1pt,
    align=center,
    inner xsep=4pt,
    inner ysep=4pt,
    font=\scriptsize,
    text width=1.65cm
  },
  commonedge/.style={MissingPairCell,line width=0.95pt},
  edgelabel/.style={
    fill=white,
    inner sep=1pt,
    font=\scriptsize,
    text=MissingPairCell
  },
  formula/.style={
    rectangle,
    draw=MatrixRule!45,
    fill=white,
    line width=0.35pt,
    rounded corners=1pt,
    align=center,
    inner xsep=4pt,
    inner ysep=3pt,
    font=\scriptsize
  }
]
  \node[coordbox] (c0) at (0,3.0)
    {coordinate \(0\)\\[2pt]
     \(A_0^{(1)}\sqcup A_0^{(2)}\sqcup A_0^{(3)}\)};
  \node[coordbox] (c1) at (5.4,3.0)
    {coordinate \(1\)\\[2pt]
     \(A_1^{(1)}\sqcup A_1^{(2)}\sqcup A_1^{(3)}\)};
  \node[coordbox] (c2) at (5.4,0)
    {coordinate \(2\)\\[2pt]
     \(A_2^{(1)}\sqcup A_2^{(2)}\sqcup A_2^{(3)}\)};
  \node[coordbox] (c3) at (0,0)
    {coordinate \(3\)\\[2pt]
     \(A_3^{(1)}\sqcup A_3^{(2)}\sqcup A_3^{(3)}\)};

  \draw[commonedge] (c0) -- node[edgelabel,above] {\(E_{01}\)} (c1);
  \draw[commonedge] (c1) -- node[edgelabel,right] {\(E_{12}\)} (c2);
  \draw[commonedge] (c2) -- node[edgelabel,below] {\(E_{23}\)} (c3);
  \draw[commonedge] (c3) -- node[edgelabel,left] {\(E_{03}\)} (c0);
  \draw[commonedge] (c0) -- node[edgelabel,sloped,above,pos=0.30] {\(E_{02}\)} (c2);
  \draw[commonedge] (c1) -- node[edgelabel,sloped,above,pos=0.70] {\(E_{13}\)} (c3);

  \node[formula] at (2.7,1.5)
    {\(\displaystyle E_{ij}=\bigcup_{r\ne s}
      A_i^{(r)}\times A_j^{(s)}\)\\[-1pt]
     for every highlighted edge \(ij\)};
\end{tikzpicture}%
}

\newcommand{\MissingCliqueDiagram}{%
\begin{tikzpicture}[
  symbol/.style={
    circle,
    draw=MatrixRule,
    fill=PresentPairCell,
    line width=0.45pt,
    minimum size=0.72cm,
    inner sep=0pt,
    font=\normalsize
  },
  coord/.style={font=\scriptsize},
  missing/.style={MissingPairCell,line width=1.05pt},
  edgelabel/.style={
    fill=white,
    inner sep=1pt,
    font=\scriptsize,
    text=MissingPairCell
  }
]
  \node[symbol] (a) at (0,2.8) {\(a\)};
  \node[symbol] (b) at (4.9,2.8) {\(b\)};
  \node[symbol] (c) at (4.9,0) {\(c\)};
  \node[symbol] (d) at (0,0) {\(d\)};

  \node[coord,anchor=south] at (a.north) {coordinate \(0\)};
  \node[coord,anchor=south] at (b.north) {coordinate \(1\)};
  \node[coord,anchor=north] at (c.south) {coordinate \(2\)};
  \node[coord,anchor=north] at (d.south) {coordinate \(3\)};

  \draw[missing] (a) -- node[edgelabel,above] {\(E_{01}\)} (b);
  \draw[missing] (b) -- node[edgelabel,right] {\(E_{12}\)} (c);
  \draw[missing] (c) -- node[edgelabel,below] {\(E_{23}\)} (d);
  \draw[missing] (d) -- node[edgelabel,left] {\(E_{03}\)} (a);
  \draw[missing] (a) -- node[edgelabel,sloped,above,pos=0.30] {\(E_{02}\)} (c);
  \draw[missing] (b) -- node[edgelabel,sloped,above,pos=0.70] {\(E_{13}\)} (d);
\end{tikzpicture}%
}

\newcommand{\LowerBoundFlowDiagram}{%
\begin{tikzpicture}[
  box/.style={
    rectangle,
    rounded corners=2pt,
    draw=MatrixRule,
    fill=PresentPairCell,
    line width=0.45pt,
    align=center,
    inner xsep=5pt,
    inner ysep=4pt,
    text width=4.25cm,
    font=\small
  },
  branch/.style={
    rectangle,
    rounded corners=2pt,
    draw=color1,
    fill=PresentPairCell,
    line width=0.6pt,
    align=center,
    inner xsep=5pt,
    inner ysep=4pt,
    text width=4.45cm,
    font=\small
  },
  conclusion/.style={
    rectangle,
    rounded corners=2pt,
    draw=FlowArrow,
    fill=PresentPairCell,
    line width=0.7pt,
    align=center,
    inner xsep=5pt,
    inner ysep=4pt,
    text width=4.75cm,
    font=\small
  },
  arrow/.style={-{Latex[length=2.4mm]},FlowArrow,line width=0.7pt},
  note/.style={font=\scriptsize,text=MatrixRule},
  secmark/.style={font=\scriptsize,text=FlowArrow,fill=white,inner sep=0.6pt}
]
  \node[box] (assume) at (0,9.2)
    {Assume a radius-two cover\\ \(C^\ast\subseteq(\Fin 8)^4\)\\ with \(|C^\ast|\le 22\)};

  \node[box] (fiber) at (0,7.55)
    {Fiber lower bounds\\ every fiber has\\ at least two codewords};

  \node[box] (missing) at (0,5.9)
    {Missing-pair graphs \(E_{ij}\)\\ row and column counts force\\ many pairs into \(E_{ij}\)};

  \node[branch] (small) at (-4.35,3.75)
    {\(|C^\ast|\le 21\) case\\ missing-neighbor intersections\\ force a missing clique};

  \node[branch] (exact) at (4.35,3.75)
    {\(|C^\ast|=22\) case\\ LRAT-certified profile classifier\\ forces a \(3+3+2\) block system};

  \node[box] (noclique) at (-4.35,1.65)
    {No missing-pair clique\\ a radius-two coverword must\\ realize one of six pairs};

  \node[box] (common) at (4.35,1.65)
    {Common-block contradiction\\ two \(3\)-symbol components\\ already require \(18\) words};

  \node[conclusion] (lower) at (0,-0.25)
    {\(K_8(4,2)\ge 23\)};

  \node[secmark,anchor=east] at ([xshift=-2pt]assume.west) {\S\ref{sec:lower-bound}};
  \node[secmark,anchor=east] at ([xshift=-2pt]fiber.west) {\S\ref{subsec:fibers-baseline}};
  \node[secmark,anchor=east] at ([xshift=-2pt]missing.west) {\S\ref{subsec:missing-pair-graphs}};
  \node[secmark,anchor=east] at ([xshift=-2pt]small.west) {\S\ref{subsec:covers-at-most-21}};
  \node[secmark,anchor=west] at ([xshift=2pt]exact.east) {\S\ref{subsec:profile-classifier}};
  \node[secmark,anchor=east] at ([xshift=-2pt]noclique.west) {\S\ref{subsec:missing-pair-graphs}};
  \node[secmark,anchor=west] at ([xshift=2pt]common.east) {\S\ref{subsec:common-block-contradiction}};
  \node[secmark,anchor=east] at ([xshift=-2pt]lower.west) {\S\ref{sec:lower-bound}};

  \draw[arrow] (assume) -- (fiber);
  \draw[arrow] (fiber) -- (missing);
  \draw[arrow] (missing.south) -- ++(0,-0.45) -| (small.north);
  \draw[arrow] (missing.south) -- ++(0,-0.45) -| (exact.north);
  \draw[arrow] (small) -- (noclique);
  \draw[arrow] (exact) -- (common);
  \coordinate (lowerLeftIn) at ([xshift=-1.05cm]lower.north);
  \coordinate (lowerRightIn) at ([xshift=1.05cm]lower.north);
  \draw[arrow] (noclique.south) -- ++(0,-0.48) -| (lowerLeftIn);
  \draw[arrow] (common.south) -- ++(0,-0.48) -| (lowerRightIn);

  \node[note] at (-2.85,4.95) {pure counting};
  \node[note] at (2.85,4.95) {finite certificate replay};
\end{tikzpicture}%
}

\hypersetup{
  hidelinks,
  colorlinks,
  breaklinks=true,
  urlcolor=color2,
  citecolor=color1,
  linkcolor=color1,
  bookmarksopen=false,
  pdftitle={A Lean-Certified Proof of K8(4,2)=23},
  pdfauthor={Andreas Florath}
}

\JournalInfo{Preprint, June 2026}
\Archive{}
\PaperTitle{A \Lean-Certified Proof of \texorpdfstring{$K_8(4,2)=23$}{K8(4,2)=23}}
\Authors{Andreas Florath\orcidlink{0009-0001-6471-7372}\textsuperscript{1}}
\affiliation{\textsuperscript{1}Deutsche Telekom AG, \textit{Andreas.Florath@telekom.de}}
\newcommand{\keywordname}{Keywords}

\Abstract{We prove the exact octonary covering-code value \(K_8(4,2)=23\)
in \Lean{} 4.  The upper bound is given by an explicit 23-word radius-two
code in \((\Fin 8)^4\), checked over all \(8^4\) ambient words.  The lower
bound excludes covers with at most 22 words.
A fiber-counting and missing-pair argument first rules out covers with at most
21 words.
In the remaining 22-word case, the proof
reduces a hypothetical cover to six missing-pair graphs coming from the
coordinate-pair projections.  Fiber-counting arguments constrain these graphs,
and two \Lean-checked Linear RAT (LRAT) refutations of
stored conjunctive-normal-form (CNF) instances force a common \(3+3+2\) block
structure.
This structure is incompatible with a 22-word cover: the two three-symbol
components already force 18 codewords, while the remaining two-symbol
component would require a binary strength-two array of length four with at
most four rows, which is impossible.  The result is packaged as a
proof-carrying \Lean{} artifact: the explicit upper bound, structural lower
bound, CNF instances, and LRAT refutations are checked inside \Lean{}, with
no external SAT solver used during proof replay.}

\begin{document}

\maketitle
\tableofcontents

\section{Introduction}

For \(N\ge1\), write \(\Fin N=\{0,1,\ldots,N-1\}\).  We use this set notation
for readability; in \Lean{}, \(\Fin N\) is the finite type of natural numbers
less than \(N\).  For a finite alphabet of size \(q\), let
\(\mathcal H_q(n)=(\Fin q)^{\Fin n}\) be the \(q\)-ary Hamming space of length
\(n\).  Thus a word is a function \(x:\Fin n\to\Fin q\), and we write \(x_i\)
for its value at coordinate \(i\).  When no ambiguity is possible, we also
write \((\Fin q)^n\) for \((\Fin q)^{\Fin n}\).  The Hamming distance
\(\distH(x,y)\) counts the coordinates on which two words differ.  A finite set
\(C\subseteq\mathcal H_q(n)\) is a radius-\(r\) covering code if every word of
\(\mathcal H_q(n)\) is within distance \(r\) of some \(c\in C\).  The covering
number \(K_q(n,r)\) is the minimum size of such a code
\cite[Sec.~2.1]{cohen1997covering}.

This paper proves one exact value:
\[
  K_8(4,2)=23.
\]
The result sits at the scale where a direct explicit upper bound is still easy
to check, but the lower bound is no longer a routine volume argument.  The main
mathematical content is a lower-bound argument excluding all 22-word covers.
Its finite classification steps are small enough to state as combinatorial
obstructions, but large enough that the proof-carrying \Lean{} formalization is
what makes the argument executable and independently checkable.

The proof is part of the \Lean{} 4 formalization of \(q\)-ary covering
codes described in~\cite{florath2026formal_foundations_covering_codes}.  That
development represents words as functions \(\Fin n\to\Fin q\), uses
finset-valued covering predicates, and states covering-number claims as
upper-bound, lower-bound, and exactness certificates rather than as a primitive
noncomputable minimum function.  Thus the exact value above is obtained from
two independent proof objects:
\[
  K_8(4,2)\le 23
  \qquad\text{and}\qquad
  K_8(4,2)\ge 23.
\]
The first proof object is an explicit code.  The second proof object is a
formal contradiction from an arbitrary assumed cover with at most 22 words.

\paragraph{Artifact availability.}
The formal artifact is identified by the source repository
{\footnotesize\begin{center}
  \artifactrepo{}
\end{center}}
at revision
{\footnotesize\begin{center}
  \artifactcommit{}.
\end{center}}
All source paths in this manuscript are relative to that repository root.  The
main declarations and certificate data are in:
\begin{center}
\footnotesize
\leanfileat{CoveringCodes/Database/Sources/OctonaryFourTwo.lean}{1}\\
\leanfileat{CoveringCodes/Database/Sources/OctonaryFourTwoBlockLRAT.lean}{1}\\
\leandir{data/K_8_4_2/lrat/}
\end{center}
The printed paths above are also hyperlinks to the fixed artifact revision.

\section{Prior Art and Scope}

The notation and basic covering-code background used here follow the standard
reference of Cohen, Honkala, Litsyn, and Lobstein~\cite{cohen1997covering}.
For the concrete parameter studied in this paper, Kéri's online tables of
bounds on covering codes list the octonary case \((n,r)=(4,2)\) with the
interval \([22,23]\)~\cite{keriCoveringCodes}.  The lower side is the Rodemich
rook-domain bound, and the upper side is covered by the large-alphabet bounds
of Kéri and Östergård~\cite{lit_rodemich_1970,lit_keri_ostergard_2005}.

This paper closes that single interval by proving the lower bound
\(K_8(4,2)\ge23\).  It does not attempt to survey or reproduce covering-code
tables more broadly.  The auxiliary \(7\)-ary radius-one fact used in
Lemma~\ref{lem:k731-lower} is a weaker formalized consequence of a classical
result: Kalbfleisch and Stanton proved the exact value
\(K_7(3,1)=25\)~\cite{lit_kalbfleisch_stanton_1969}, while this paper only
needs and formalizes the bound \(K_7(3,1)\ge22\).  The main parameter-specific
contribution is the structural argument and the finite 22-word profile
classifier for \(K_8(4,2)\).

The formal development is written in the \Lean{} 4 theorem prover and uses
\mathlib{} as its mathematical library~\cite{demoura2021lean4,mathlib2020}.
The formal-foundations paper \emph{Formal Foundations and Proof-Carrying
Certificates for \(q\)-ary Covering Codes in Lean 4} supplies the general \Lean{}
infrastructure for \(q\)-ary Hamming spaces, covering-code predicates, and
upper/lower/exactness certificates~\cite{florath2026formal_foundations_covering_codes}.
The present paper contributes the octonary proof object on top of that
infrastructure.  Its finite classifier uses Linear RAT (LRAT) certificates,
following the proof-checking format introduced by Cruz-Filipe, Heule, Hunt,
Kaufmann, and Schneider-Kamp for efficient certified RAT
verification~\cite{cruzfilipe2017lrat}.  An LRAT refutation is a clausal proof
that derives the empty clause from a CNF formula; its proof steps carry enough
checking information for \Lean{} to replay the contradiction mechanically.

\section{Contributions}

The contributions of this paper are the following.

\begin{itemize}[leftmargin=*]
\item It proves the exact covering-code value \(K_8(4,2)=23\), closing the
  previously recorded interval \([22,23]\).  The lower-bound proof excludes
  all 22-word covers by reducing hypothetical small covers to fiber counts and
  missing coordinate-pair graphs, separating the case \(|C|\le21\),
  classifying the remaining 22-word profiles, and deriving a common-block
  contradiction.
\item It supplies a proof-carrying \Lean{} artifact for the result.  The
  explicit upper bound, structural lower-bound reductions, finite CNF
  classifiers, and LRAT refutations are all connected to the same
  upper/lower/exactness certificate interface introduced in the
  formal-foundations paper.
\item It records a reusable pattern for attaching solver-produced finite
  evidence to covering-code lower bounds: the solver output is untrusted, while
  \Lean{} checks the certificate data and the mathematical bridge from the
  finite classifier back to the covering-code statement.
\end{itemize}

\section{The Upper Bound}

The upper bound is represented by the following 23 codewords, written as
four-digit words over \(\{0,\ldots,7\}\):
\[
\begin{array}{rrrrrrrr}
0000 & 4465 & 4576 & 7607 & 7710 & 7021 & 1132 & 1233 \\
6564 & 6375 & 6456 & 5017 & 5620 & 5701 & 2232 & 2143 \\
1242 & 4354 & 3555 & 3474 & 3366 & 0727 & 0611
\end{array}
\]
This code was found by local search and then translated by coordinate and alphabet
normalization so that \(0000\) is one of the centers.

\begin{proposition}[Explicit upper certificate]\label{prop:k842-upper}
The displayed code proves \(K_8(4,2)\le23\).
\end{proposition}

\begin{proof}
Let \(C\) be the displayed list of codewords.  The formal upper-bound
certificate requires two finite checks.  First, the list has cardinality at most
23:
\[
  |C|\le 23.
\]
Although the printed list contains 23 distinct words, the certificate interface
only needs this upper-cardinality statement.  Second, every word of
\((\Fin 8)^4\) is within distance two of one listed word:
\[
  \forall x:\Fin 4\to\Fin 8,\;
  \exists c\in C,\; \distH(x,c)\le 2.
\]
This is an exhaustive finite statement over \(8^4=4096\) ambient words.  Once
the list is fixed, no optimization or search remains in the upper-bound
argument.

\LeanRefs{\leanmaindecl{octonaryFourRadiusTwoCode}{36};
\leanmaindecl{octonaryFourRadiusTwo_card}{61};
\leanmaindecl{octonaryFourRadiusTwo_covers}{66};
\leanmaindecl{octonaryFourRadiusTwoExplicit}{75}.}
\end{proof}

\section{The Lower Bound}\label{sec:lower-bound}

\begin{figure*}[t]
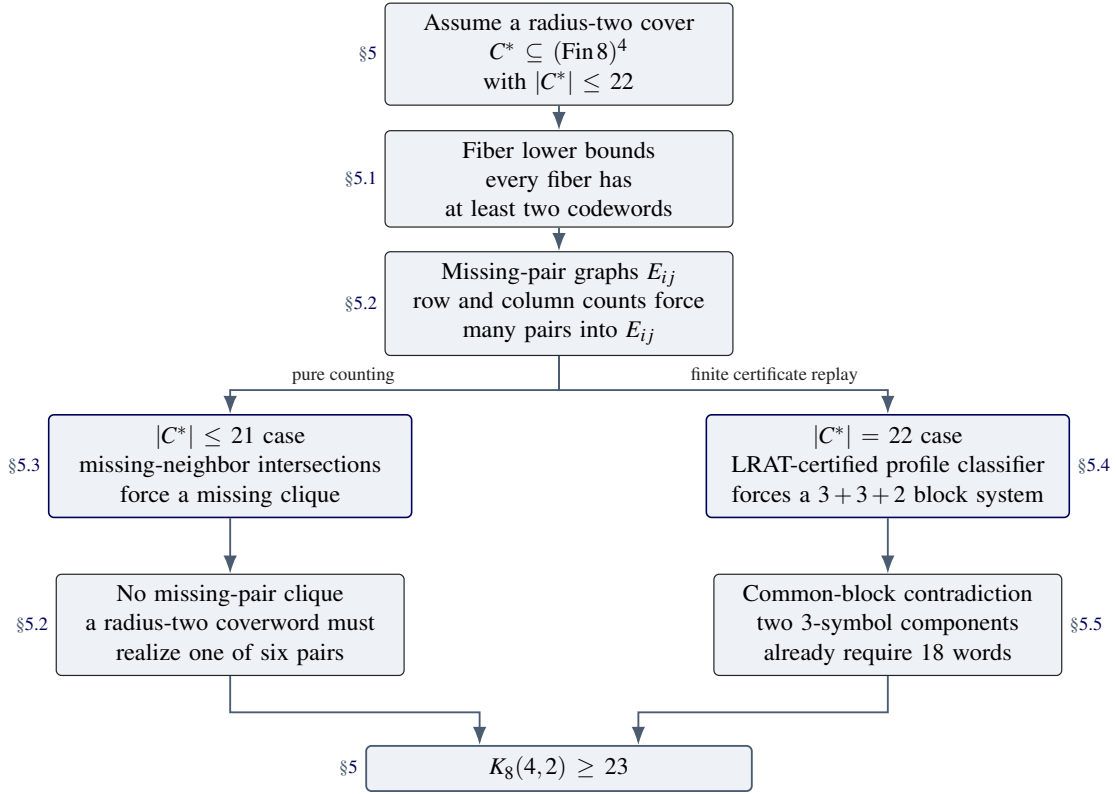

\centering
\LowerBoundFlowDiagram
\caption{Proof flow for the lower bound \(K_8(4,2)\ge23\).  The
\(|C^\ast|\le21\) branch is discharged by counting in the missing-pair graphs;
the exact \(22\)-word branch uses the LRAT-certified profile classifier before
the common-block contradiction.}
\label{fig:lower-bound-flow}
\end{figure*}

The lower bound is the part of the result that requires the main argument.  It
is not a volume bound; it is a structural nonexistence argument for 22-word
covers.  The proof replaces a
hypothetical code by a small profile of fiber sizes and missing coordinate
pairs, proves finite classification statements about those profiles, and then
derives a block-structure contradiction.  We present the argument as a sequence
of elementary reductions followed by two finite profile-classification
statements.

The section is deliberately organized as several small lemmas rather than one
large nonexistence theorem.  Each lemma isolates one property of a hypothetical
cover: a local covering consequence, a fiber-size restriction, a missing-pair
criterion, or a finite profile exclusion.  This makes the proof easier to
follow, and it also reflects the formal development, where each property is a
separate theorem that can be reused by later steps.  Some auxiliary statements
used below are weaker than the best known informal results; they are included in
this form because these are the bounds currently formalized in \Lean.
Figure~\ref{fig:lower-bound-flow} summarizes the resulting proof flow.

\paragraph{Hypothetical counterexample.}
Assume, for contradiction, that there is a radius-two cover
\[
  C^\ast\subseteq(\Fin 8)^4,\qquad |C^\ast|\le22.
\]
The radius-two covering condition is the concrete statement
\[
  \forall x\in(\Fin 8)^4,\; \exists c\in C^\ast,\; \distH(x,c)\le2.
\]
The symbol \(C^\ast\) is fixed by this assumption.  The rest of this section
derives a contradiction from it.

\subsection{Fibers and baseline counts}\label{subsec:fibers-baseline}

For a coordinate \(i\in\Fin 4\) and symbol \(a\in\Fin 8\), write
\[
  C^\ast_{i,a}=\{\,c\in C^\ast : c_i=a\,\}.
\]
We call \(C^\ast_{i,a}\) the fiber of \(C^\ast\) over \((i,a)\).

For example, if \(u=(1,4,5,7)\), \(v=(0,1,5,1)\), and
\(w=(1,3,2,0)\) are codewords in \(C^\ast\), then
\(u,w\in C^\ast_{0,1}\), because their zeroth coordinate is \(1\), and
\(u,v\in C^\ast_{2,5}\), because their second coordinate is \(5\).

For each fixed coordinate \(i\), these eight fibers partition \(C^\ast\):
every codeword has exactly one value in coordinate \(i\).  Thus
\[
  C^\ast=\bigsqcup_{a\in\Fin 8} C^\ast_{i,a},
  \qquad
  |C^\ast|=\sum_{a\in\Fin 8}|C^\ast_{i,a}|.
\]
Here \(\bigsqcup\) denotes a disjoint union; the displayed decomposition is a
partition into pairwise disjoint fibers.
For the same coordinate \(i\) and symbol \(a\), write
\[
  S_{i,a}=\{\,x\in(\Fin 8)^4 : x_i=a\,\}.
\]
Thus \(S_{i,a}\) is the coordinate slice of \((\Fin 8)^4\) where coordinate
\(i\) is fixed to \(a\).  The other three coordinates are free, and deleting
coordinate \(i\) gives a Hamming-distance-preserving bijection from
\(S_{i,a}\) to \((\Fin 8)^3\).
For every \(i\in\Fin 4\) and \(a\in\Fin 8\), the radius-two covering condition
above implies that every point of \(S_{i,a}\) is within distance at most two of
some codeword of \(C^\ast\).

\begin{lemma}[Missing tail symbols force a tail pair]\label{lem:k731-forced-pair}
Let \(C\subseteq(\Fin 7)^3\) be a radius-one cover.  Fix a coordinate
\(j\in\Fin 3\) and write
\[
  \pi_j:(\Fin 7)^3\to(\Fin 7)^2
\]
for the projection that deletes coordinate \(j\), keeping the remaining
coordinates in their original order; for example,
\(\pi_1(x_0,x_1,x_2)=(x_0,x_2)\).  Fix also a symbol \(a\in\Fin 7\), and put
\[
  F_{j,a}=\{\,c\in C:c_j=a\,\}.
\]
Thus \(F_{j,a}\) is the analogue of the fibers \(C^\ast_{i,a}\) above:
it is the part of \(C\) whose \(j\)-th coordinate is \(a\).
Let
\begin{align*}
  B_0&=\Fin 7\setminus\{\,(\pi_j c)_0:c\in F_{j,a}\,\},\\
  B_1&=\Fin 7\setminus\{\,(\pi_j c)_1:c\in F_{j,a}\,\}.
\end{align*}
Thus \(B_0\) and \(B_1\) record the symbols that do not occur in the
first and second remaining coordinate, respectively, among the projected
fiber \(\pi_j(F_{j,a})\).
Then
\[
  B_0\times B_1
  \subseteq
  \{\,\pi_j(e):e\in C\,\}.
\]
Informally, even though \(b_0\) and \(b_1\) are absent from the two
projected coordinates of the fiber \(F_{j,a}\), the pair \((b_0,b_1)\)
must occur as the projection of some codeword of \(C\).
\end{lemma}

\begin{proof}
To prove the inclusion, take an arbitrary pair
\((b_0,b_1)\in B_0\times B_1\).  Since \(\pi_j\) deletes only coordinate
\(j\), the value \(a\) in coordinate \(j\) together with the projected pair
\((b_0,b_1)\) determines a unique word \(x\in(\Fin 7)^3\):
\[
  x_j=a,\qquad \pi_j(x)=(b_0,b_1).
\]
For instance, if \(j=1\), then this word is \(x=(b_0,a,b_1)\).
Since \(C\) is a radius-one cover, there is \(e\in C\) with
\(\distH(x,e)\le1\).

If \(e_j=a\), then \(e\in F_{j,a}\).  By the definitions of \(B_0\) and \(B_1\),
the two tail coordinates of \(\pi_j(e)\) are different from \(b_0\) and \(b_1\),
respectively.  Thus \(x\) and \(e\) differ in both tail coordinates, so
\(\distH(x,e)\ge2\), a contradiction.

Therefore \(e_j\ne a\).  Thus \(x\) and \(e\) already disagree in coordinate
\(j\), the coordinate deleted by \(\pi_j\).  Since \(\distH(x,e)\le1\), there can be
no disagreement in the two tail coordinates.  Hence \(\pi_j(e)=(b_0,b_1)\),
proving that \((b_0,b_1)\in\{\,\pi_j(e):e\in C\,\}\).

\LeanRefs{\leansparsedecl{forced_tail_pair_subset_of_fiber_cover}{282}.}
\end{proof}

\begin{lemma}[Auxiliary \(7\)-ary lower bound]\label{lem:k731-lower}
\[
  K_7(3,1)\ge22.
\]
\end{lemma}

\begin{proof}
This is weaker than the classical value \(K_7(3,1)=25\) proved by
Kalbfleisch and Stanton~\cite{lit_kalbfleisch_stanton_1969}.  The formalized
argument for the weaker bound is the following counting proof.

Suppose, for contradiction, that \(C\subseteq(\Fin 7)^3\) is a radius-one cover
with \(|C|\le21\).  For a coordinate \(j\) and a symbol \(a\), define
\[
  F_{j,a}=\{\,c\in C:c_j=a\,\}.
\]
We first show that no such fiber can have size at most two.  Indeed, suppose
\(|F_{j,a}|\le2\).  Apply Lemma~\ref{lem:k731-forced-pair} to this \(j\) and
\(a\), using the projection \(\pi_j\) that deletes coordinate \(j\), and let
\(B_0,B_1\) be the missing-symbol sets defined there.  Fix one
of the two coordinates different from \(j\).  Each codeword in \(F_{j,a}\) has
exactly one symbol in that coordinate, so the words of \(F_{j,a}\) use at most
two symbols there.  Since the alphabet is \(\Fin 7\), at least \(7-2=5\)
symbols are missing in that coordinate.  This applies to both coordinates
different from \(j\), so \(|B_0|\ge5\) and \(|B_1|\ge5\).  The lemma then says
that every pair in \(B_0\times B_1\) occurs as the \(\pi_j\)-projection of some
codeword.  Thus at least
\(5\cdot5=25\) projected pairs occur.  This is impossible when \(|C|\le21\),
since each codeword contributes only one such pair.

Hence every fiber \(F_{j,a}\) has size at least three.  Since the seven
fibers in any fixed coordinate partition \(C\), the assumption \(|C|\le21\)
forces every such fiber to have size exactly three.

Now fix arbitrary symbols \(b,c\in\Fin 7\).  We show that some codeword of
\(C\) has second coordinate \(b\) and third coordinate \(c\).  The set of
first-coordinate symbols occurring in words with second coordinate \(b\) has
size at most three, and the set of first-coordinate symbols occurring in words
with third coordinate \(c\) also has size at most three.  The union of these
two sets has size at most six, so choose a first-coordinate symbol \(a\)
outside this union.  Then, in the fiber \(F_{0,a}\), the symbol \(b\) is
missing from coordinate \(1\), and the symbol \(c\) is missing from coordinate
\(2\).  Applying Lemma~\ref{lem:k731-forced-pair} with \(j=0\), the projection
\(\pi_0\) deletes the first coordinate and keeps coordinates \(1\) and \(2\).
Thus \(b\in B_0\) and \(c\in B_1\), so the lemma shows that \((b,c)\) occurs
as the last-two-coordinate projection of some codeword.  Since
\((b,c)\in(\Fin 7)^2\) was arbitrary, all \(7^2=49\) last-two-coordinate pairs
occur.  On the other hand, the set \(\{\,\pi_0(e):e\in C\,\}\) is the image of
\(C\) under the map \(e\mapsto\pi_0(e)\), so it has at most \(|C|\le21\)
elements.  This contradicts the \(49\) distinct pairs just obtained, and proves
\(K_7(3,1)\ge22\).

\LeanRefs{\leansparsedecl{qarySevenThreeOneLowerTwentyTwo}{400}, of type
\(\texttt{QaryKLower 7 3 1 22}\).}
\end{proof}

\begin{lemma}[Fiber lower bound]\label{lem:fiber-lower}
For every coordinate \(i\in\Fin 4\) and every symbol \(a\in\Fin 8\),
\[
  |C^\ast_{i,a}|\ge2.
\]
\end{lemma}

\begin{proof}
Fix \(i\) and \(a\), and consider the coordinate layer \(S_{i,a}\).
Deleting coordinate \(i\) gives a bijection \(S_{i,a}\cong(\Fin 8)^3\).

Suppose first that \(C^\ast_{i,a}=\emptyset\).  Every codeword in \(C^\ast\)
then disagrees in coordinate \(i\) with every point of \(S_{i,a}\).  Hence, if
a codeword \(c\in C^\ast\) covers a point \(x\in S_{i,a}\) at radius two, then
\(c\) and \(x\) already differ in coordinate \(i\).  After deleting coordinate
\(i\), their projections differ in at most one remaining coordinate.  Thus the
projections of the words of \(C^\ast\) form a radius-one cover of
\((\Fin 8)^3\) with at most 22 centers.  This is impossible by the
volume bound, since a radius-one ball in \((\Fin 8)^3\) has
\[
  1+3(8-1)=22
\]
points.  The first \(22\) in the following product is the assumed upper bound
on the number of centers; the second \(22\) is the size of one radius-one ball.
Hence, even before accounting for overlaps, these balls contain at most
\[
  22\cdot22=484<512=8^3=|(\Fin 8)^3|
\]
points.  This contradicts the projected radius-one cover, and therefore
\(C^\ast_{i,a}\ne\emptyset\).

It remains to exclude \(|C^\ast_{i,a}|=1\).  Let \(C^\ast_{i,a}=\{c\}\), and let \(S\)
be the subcube of \(S_{i,a}\) in which each remaining coordinate avoids the
corresponding symbol of \(c\):
\[
  S=\{\,x\in(\Fin 8)^4 :
      x_i=a\text{ and }x_j\ne c_j\text{ for all }j\ne i\,\}.
\]
Thus \(S\) consists of the words in the slice \(x_i=a\) that avoid matching
the codeword \(c\) in every other coordinate.  For example, if \(i=0\) and
\(c=(a,c_1,c_2,c_3)\), then
\[
  S=\{\,(a,x_1,x_2,x_3):x_1\ne c_1,\ x_2\ne c_2,\ x_3\ne c_3\,\}.
\]
In particular, \(S\) is nonempty: each of the three free coordinates has
\(7\) choices.
After deleting coordinate \(i\), each remaining coordinate \(j\ne i\) has the
seven allowed symbols \(\Fin 8\setminus\{c_j\}\).  These three \(7\)-element
sets need not be the same.  Choosing a bijection from each of them to \(\Fin 7\)
identifies \(S\) with \((\Fin 7)^3\), and this coordinatewise
relabeling
preserves Hamming distance.

The word \(c\) agrees with every \(x\in S\) in coordinate \(i\), since
\(c_i=a=x_i\).  In every other coordinate \(j\ne i\), the definition of \(S\)
forces \(x_j\ne c_j\).  Therefore \(\distH(x,c)=3\) for all \(x\in S\), so
\(c\) covers no point of \(S\) at radius two.

Since \(C^\ast\) covers \(S\) and \(c\) covers no point of \(S\), every
\(x\in S\) is covered by some \(d\in C^\ast\setminus\{c\}\).  Since \(c\) is
the only word in the fiber \(C^\ast_{i,a}\), this \(d\) satisfies \(d_i\ne a\).
But \(x_i=a\), so \(d\) already disagrees with \(x\) in coordinate \(i\).
Because \(d\) covers \(x\) at radius two in the original four coordinates,
there can be at most one further disagreement among the remaining three
coordinates.  Consequently, after deleting coordinate \(i\), the projections of
the words in \(C^\ast\setminus\{c\}\) cover the projected copy of \(S\) at
radius one.  However, these projected words, viewed as radius-one centers, still
lie in the ambient space \((\Fin 8)^3\).  For example, if after deleting
coordinate \(i\) the forbidden symbols are \(c_1,c_2,c_3\), then the projected
copy of \(S\) consists of the triples \((y_1,y_2,y_3)\) with
\(y_k\ne c_k\) for each \(k\).  A projected codeword may still have one of
these forbidden symbols as a coordinate, and therefore may lie outside the
projected copy of \(S\).  The next step replaces each ambient radius-one center
by a center inside the projected copy of \(S\), in such a way that every point
of the projected copy of \(S\) covered by the original center is still covered
after the replacement.

Let \(z=\pi_i(d)\) be the projection of one of the words
\(d\in C^\ast\setminus\{c\}\).  Thus \(z\in(\Fin 8)^3\) is a radius-one
center for the projected problem, although it need not lie in the projected
copy of \(S\).  We use the following observation, applied to each such
projected center \(z\).  The intersection of its ambient radius-one ball with
the \(7\)-ary subcube is either empty or contained in a radius-one ball whose
center lies in the subcube.  Indeed, if \(z\) has no excluded coordinate, then
\(z\) already lies in the subcube.  If \(z\) has exactly one excluded
coordinate \(k\) and the intersection is nonempty, every point \(y\) in this
intersection must differ from \(z\) in coordinate \(k\), because \(y_k\) is
allowed while \(z_k\) is excluded.  Since \(y\) is within radius one of \(z\),
it therefore agrees with \(z\) in the other two coordinates.  Replacing
\(z_k\) by any allowed symbol gives a center in the subcube whose radius-one
ball still contains every such \(y\).  Finally, if \(z\) has at least two
excluded coordinates, every point of the subcube differs from \(z\) in at
least two coordinates, so the intersection is empty.

Replacing the projected ambient centers in this way gives a radius-one cover of
the projected copy of \(S\) whose centers lie inside the \(7\)-ary subcube.
Some moved centers may coincide, and centers with empty intersection may be
discarded, but this can only decrease the number of centers.  Therefore the new
cover still has at most \(21\) centers, since \(|C^\ast|\le22\) and the word
\(c\) covers no point of \(S\).

But, after the coordinatewise relabeling above, the projected copy of \(S\) is
the Hamming space \((\Fin 7)^3\).  Thus such a cover would be a radius-one
cover of a \(7\)-ary Hamming space of length three with at most \(21\) centers.
This contradicts Lemma~\ref{lem:k731-lower}.  Hence the assumed case
\(|C^\ast_{i,a}|\le1\) is impossible, and every fiber \(C^\ast_{i,a}\) has at
least two codewords.

\LeanRefs{\leanmaindecl{octonaryFourRadiusTwo_symbol_fiber_card_ge_two_of_card_le_twenty_two}{374};
\leanmaindecl{octonaryFourRadiusTwo_no_singleton_symbol_of_card_le_twenty_two}{358};
\leansparsedecl{qarySevenThreeOneLowerTwentyTwo}{400}.}
\end{proof}

\begin{definition}[Fiber excess above the baseline]\label{def:fiber-excess}
For the fixed hypothetical cover \(C^\ast\), set \(m=|C^\ast|\).  For each
coordinate \(i\in\Fin 4\), define the excess function
\[
  \varepsilon_i:\Fin 8\to\mathbb{N},
  \qquad
  \varepsilon_i(a)=|C^\ast_{i,a}|-2.
\]
\end{definition}

This notation is attached only to the fixed cover \(C^\ast\).  The number
\(\varepsilon_i(a)\) measures how many codewords lie in the fiber
\(C^\ast_{i,a}\) beyond the baseline of two codewords forced by
Lemma~\ref{lem:fiber-lower}.  Thus the quantities \(\varepsilon_i(a)\) are
nonnegative.  For every coordinate \(i\),
\[
  \sum_{a\in\Fin 8}\varepsilon_i(a)=m-16.
\]
For a fixed coordinate \(i\), the eight fibers \(C^\ast_{i,a}\), \(a\in\Fin 8\),
partition \(C^\ast\).  Lemma~\ref{lem:fiber-lower} gives the baseline
contribution \(8\cdot2=16\), so the sum of the excess values is precisely the
amount by which the eight fiber sizes exceed this forced baseline.  Since
\(m\le22\), this total excess is at most six; it is exactly six when \(m=22\),
and at most five when \(m\le21\).

\begin{center}
  \captionsetup{type=figure,hypcap=false,font=footnotesize,skip=2pt}
  \FiberCountDiagram
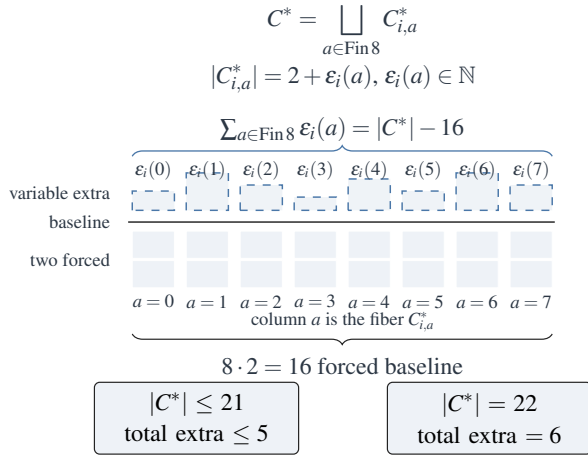
  \caption{Fiber counting for one fixed coordinate \(i\): columns are the
  fibers \(C^\ast_{i,a}\), the two lower boxes are forced by
  Lemma~\ref{lem:fiber-lower}, and dashed regions are schematic extras
  \(\varepsilon_i(a)\).}
  \label{fig:fiber-counts}
\end{center}

\subsection{Missing pair graphs}\label{subsec:missing-pair-graphs}

\begin{definition}[Missing-pair profile]\label{def:missing-profile}
For each coordinate pair \(0\le i<j<4\), define the missing-pair graph
\(E_{ij}\) as follows.  Its left vertices are the eight symbols in coordinate
\(i\), and its right vertices are the eight symbols in coordinate \(j\).  An
edge from the left vertex \(a\in\Fin 8\) to the right vertex \(b\in\Fin 8\)
records that the coordinate pair \((a,b)\) is missing from \(C^\ast\).  Thus we
identify the edge set with a subset
\[
  E_{ij}\subseteq \Fin 8\times\Fin 8
\]
defined by
\[
  (a,b)\in E_{ij}
  \quad\Longleftrightarrow\quad
  \text{no }c\in C^\ast\text{ satisfies }c_i=a,\ c_j=b.
\]
\end{definition}

The fiber bounds enter these graphs through their \(8\times8\)
adjacency-matrix view.  For \(E_{ij}\), rows are indexed by symbols in
coordinate \(i\), columns are indexed by symbols in coordinate \(j\), and the
entry \((a,b)\) is marked missing exactly when \((a,b)\in E_{ij}\).  Fix
\(i<j\).  The row indexed by \(a\in\Fin 8\) can be filled only by codewords in
the fiber
\[
  C^\ast_{i,a}=\{\,c\in C^\ast:c_i=a\,\},
\]
because every pair in that row has first coordinate \(a\).  Each such codeword
realizes at most one entry in the row, namely the entry in column \(c_j\).
Consequently the row at \(a\) has at least
\[
  8-|C^\ast_{i,a}| = 6-\varepsilon_i(a)
\]
entries belonging to \(E_{ij}\).  The same argument with the two coordinates exchanged shows
that the column at \(b\in\Fin 8\) has at least
\[
  8-|C^\ast_{j,b}| = 6-\varepsilon_j(b)
\]
entries belonging to \(E_{ij}\).

For \(i<j\) and \(u\in\Fin 8\), write
\[
  N_{ij}(u)=\{\,v\in\Fin 8 : (u,v)\in E_{ij}\,\}.
\]
Thus \(N_{ij}(u)\) is the set of symbols in coordinate \(j\) that form a
missing pair with the symbol \(u\) in coordinate \(i\), or equivalently the set
of column symbols \(v\) for which the row entry \((u,v)\) belongs to
\(E_{ij}\) in the adjacency-matrix view.
Only codewords in \(C^\ast_{i,u}\) can make entries in this row present, and
each such codeword can make at most one column present.  Since
\(|C^\ast_{i,u}|=2+\varepsilon_i(u)\), at least
\[
  8-(2+\varepsilon_i(u))=6-\varepsilon_i(u)
\]
column symbols \(v\) remain with \((u,v)\in E_{ij}\).  Thus
\begin{equation}\label{eq:missing-neighbor-row-bound}
  |N_{ij}(u)|\ge 6-\varepsilon_i(u).
\end{equation}

\begin{lemma}[No missing-pair clique]\label{lem:missing-clique}
The six graphs \(E_{ij}\) contain no simultaneous four-partite clique.  That is,
there are no \(a,b,c,d\in\Fin 8\) such that
\[
  (a,b)\in E_{01},\quad (a,c)\in E_{02},\quad (a,d)\in E_{03},
\]
and
\[
  (b,c)\in E_{12},\quad (b,d)\in E_{13},\quad (c,d)\in E_{23}.
\]
\end{lemma}

\begin{center}
  \captionsetup{type=figure,hypcap=false,font=footnotesize,skip=2pt}
  \MissingCliqueDiagram
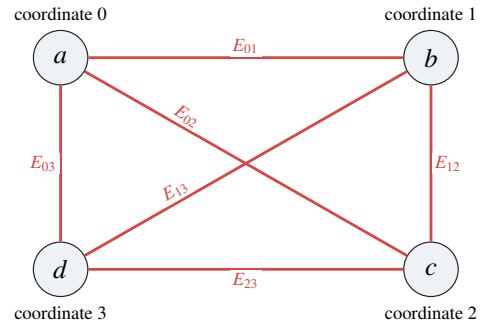
  \caption{Forbidden configuration in Lemma~\ref{lem:missing-clique}: if all
  six coordinate pairs of \(x=(a,b,c,d)\) lie in the corresponding \(E_{ij}\),
  then no radius-two coverword can cover \(x\).}
  \label{fig:missing-clique}
\end{center}

\begin{proof}
Suppose such \(a,b,c,d\) exist, and set \(x=(a,b,c,d)\).  The assumptions say
that all six two-coordinate projections of \(x\) are missing from \(C^\ast\).
For instance, \((a,b)\in E_{01}\) says that no codeword of \(C^\ast\) has
coordinates \(0,1\) equal to \((a,b)\), and \((a,c)\in E_{02}\) says the same
for coordinates \(0,2\).

Since \(C^\ast\) is a radius-two cover, there is a codeword \(e\in C^\ast\)
with \(\distH(x,e)\le2\).  The words \(x\) and \(e\) have length four, so
distance at most two means that they disagree in at most two coordinates.
Equivalently, they agree in at least two coordinates.  Among the four
coordinate positions \(0,1,2,3\), any two agreement positions are one of
\[
  \{0,1\},\{0,2\},\{0,3\},\{1,2\},\{1,3\},\{2,3\}.
\]
These are exactly the six coordinate pairs whose missing-pair graphs are
\(E_{01},E_{02},E_{03},E_{12},E_{13},E_{23}\).  If, for example, the agreement
coordinates are \(0\) and \(2\),
then \(e_0=a\) and \(e_2=c\).  Thus the pair \((a,c)\) occurs in \(C^\ast\) in
coordinates \(0,2\), contradicting \((a,c)\in E_{02}\).  The other five
coordinate pairs give the same contradiction with the corresponding assumed
edge of the clique.

\LeanRefs{\leanmaindecl{no_pairMissing4_clique_of_covers}{4822};
\leanmaindecl{dist_le_two_agrees_pair4}{4794}, the formal
distance-at-most-two agreement step.}
\end{proof}

\subsection{Covers with at most 21 words}\label{subsec:covers-at-most-21}

\begin{proposition}[No cover with at most 21 words]\label{prop:no-five-excess}
The fixed cover \(C^\ast\) cannot satisfy \(|C^\ast|\le21\).
\end{proposition}

\begin{proof}
The idea is to use the small budget above the baseline to build a forbidden
configuration in the missing-pair graphs.  With at most \(21\) words, each
coordinate has total excess at most five beyond the mandatory baseline of two
words in each fiber.  A row indexed by a symbol \(a\) with
\(\varepsilon_i(a)=0\) then has at least six neighbors in the corresponding
missing-pair graph, and later a row indexed by a symbol \(c\) with
\(\varepsilon_i(c)\le1\) still has at least five such neighbors.  Because all
neighbor sets live in the same eight-symbol alphabet, sets of sizes six, six,
and five are forced to overlap in the way needed below.  We choose symbols
\(a,b,c,d\) successively from these intersections so that every one of the six
coordinate pairs between them is missing.  This gives the simultaneous
four-partite clique excluded by Lemma~\ref{lem:missing-clique}.

Assume \(|C^\ast|\le21\).  Fix a coordinate \(i\).  By the fiber partition
property and the identity \(|C^\ast_{i,a}|=2+\varepsilon_i(a)\),
\[
  |C^\ast|
  =
  \sum_{a\in\Fin 8}|C^\ast_{i,a}|
  =
  16+\sum_{a\in\Fin 8}\varepsilon_i(a).
\]
Thus \(16\) is the forced baseline for one fixed coordinate: two codewords in
each of its eight fibers, as visualized in Figure~\ref{fig:fiber-counts}.
Since the present case assumes
\(|C^\ast|\le21\), only \(21-16=5\) units of excess can be distributed
among those fibers.  Thus, for this fixed coordinate \(i\),
\begin{equation}\label{eq:small-cover-total-excess}
  \sum_{a\in\Fin 8}\varepsilon_i(a)\le5.
\end{equation}
The coordinate \(i\) was arbitrary, so we may use this bound for each
coordinate \(i\in\Fin 4\).
Applying this to coordinate \(0\), consider the eight values
\(\varepsilon_0(a)\) with \(a\in\Fin 8\).  They are nonnegative by
Lemma~\ref{lem:fiber-lower}, and their sum is at most five.  Therefore at most
five of them can be positive, so at least three of the eight values are zero.
Choose \(a\in\Fin 8\) with \(\varepsilon_0(a)=0\).  Applying the neighbor
bound~\eqref{eq:missing-neighbor-row-bound} with \(i=0\) and \(u=a\) gives
\(|N_{0j}(a)|\ge 6-\varepsilon_0(a)=6\) for \(j=1,2,3\).  Hence
\[
  |N_{01}(a)|\ge6,\qquad |N_{02}(a)|\ge6,\qquad |N_{03}(a)|\ge6.
\]

The same count in coordinate \(1\) shows that at most five symbols have positive
\(\varepsilon_1\)-value.  Because \(|N_{01}(a)|\ge6\), some
\(b\in N_{01}(a)\) has \(\varepsilon_1(b)=0\).  The row bounds in \(E_{12}\) and
\(E_{13}\) give
\[
  |N_{12}(b)|\ge6,\qquad |N_{13}(b)|\ge6.
\]

We use the elementary estimate
\[
  |A\cap B|=|A|+|B|-|A\cup B|\ge |A|+|B|-8
\]
for subsets \(A,B\subseteq\Fin 8\).  Since
\(|N_{02}(a)|\ge6\) and \(|N_{12}(b)|\ge6\), this gives
\[
  |N_{02}(a)\cap N_{12}(b)|\ge 6+6-8=4.
\]
We now need a choice of \(c\) in this intersection with
\(\varepsilon_2(c)\le1\), because this will make the row \(N_{23}(c)\) large
enough for the final intersection step.
By the total-excess bound~\eqref{eq:small-cover-total-excess}, applied with
\(i=2\), the total amount above the baseline in coordinate \(2\) is at most
five.  If every element of this intersection had \(\varepsilon_2\)-value at
least \(2\), then the total excess in coordinate \(2\) would be at least
\(4\cdot2=8\), contradicting
\(\sum_{u\in\Fin 8}\varepsilon_2(u)\le5\).  Hence one can choose
\[
  c\in N_{02}(a)\cap N_{12}(b)
  \qquad\text{with}\qquad
  \varepsilon_2(c)\le1.
\]
Now apply the neighbor bound~\eqref{eq:missing-neighbor-row-bound} to
\(E_{23}\), with \(i=2\), \(j=3\), and \(u=c\).  The rows of \(E_{23}\) are
indexed by coordinate-\(2\) symbols, so the row \(c\) records the
coordinate-\(3\) symbols \(d\) for which \((c,d)\) is missing.  Since
\(\varepsilon_2(c)\le1\), the bound gives
\[
  |N_{23}(c)|\ge 6-\varepsilon_2(c)\ge5.
\]

Applying the same estimate to \(N_{03}(a)\) and \(N_{13}(b)\) gives
\[
  |N_{03}(a)\cap N_{13}(b)|\ge 6+6-8=4.
\]
Applying it once more to
\[
  N_{03}(a)\cap N_{13}(b)
  \qquad\text{and}\qquad
  N_{23}(c)
\]
gives
\[
  |N_{03}(a)\cap N_{13}(b)\cap N_{23}(c)|
  \ge 4+5-8=1.
\]
Thus this triple intersection is nonempty.  Choose \(d\) in it.

By the definition of the
sets \(N_{ij}\), the earlier choices give
\[
  (a,b)\in E_{01},\qquad
  (a,c)\in E_{02},\qquad
  (b,c)\in E_{12},
\]
and the choice of \(d\) gives
\[
  (a,d)\in E_{03},\qquad
  (b,d)\in E_{13},\qquad
  (c,d)\in E_{23}.
\]
Hence all six pairs
\[
  (a,b),\ (a,c),\ (a,d),\ (b,c),\ (b,d),\ (c,d)
\]
are missing in the corresponding graphs.  Thus
\((a,b,c,d)\) is a simultaneous four-partite clique in the six graphs
\(E_{ij}\), contradicting Lemma~\ref{lem:missing-clique}.  Therefore the
assumption \(|C^\ast|\le21\) is impossible; any radius-two cover in
\((\Fin 8)^4\) must have at least \(22\) codewords.

The formal lower-bound theorem has target
\[
  \texttt{QaryKLower 8 4 2 22}.
\]
\LeanRefs{\leanmaindecl{octonaryFourRadiusTwoLowerTwentyTwo}{5455};
\leanmaindecl{exists_four_partite_clique_of_deficit_budget_five}{656}.}
\end{proof}

\subsection{The 22-word profile classifier}\label{subsec:profile-classifier}

\begin{proposition}[22-word profile classifier]\label{prop:profile-classifier}
If \(|C^\ast|=22\), then, after relabeling the alphabet independently in each
coordinate, there are partitions
\[
  A_i^{(1)}\sqcup A_i^{(2)}\sqcup A_i^{(3)}=\Fin 8
  \qquad (i=0,1,2,3)
\]
with \(|A_i^{(1)}|=|A_i^{(2)}|=3\) and \(|A_i^{(3)}|=2\), such that for every
coordinate pair \(i<j\),
\[
  E_{ij}=\bigcup_{r\ne s} A_i^{(r)}\times A_j^{(s)}.
\]
\end{proposition}

\begin{center}
  \captionsetup{type=figure,hypcap=false,font=footnotesize,skip=2pt}
  \MissingPairBlockMatrix
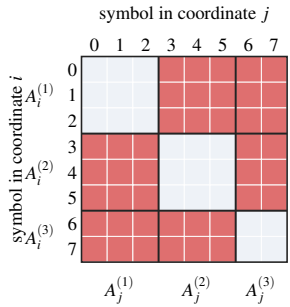
  \caption{Block form of \(E_{ij}\) for \(i<j\): red off-diagonal cells are
  the missing pairs in \(E_{ij}\), while gray diagonal cells are present pairs.}
  \label{fig:missing-pair-block-form}
\end{center}

\begin{proof}[Proof sketch]
Let \(C^\ast\) be a 22-word cover.  For each coordinate \(i\), the eight
fibers \(C^\ast_{i,a}\), \(a\in\Fin 8\), partition \(C^\ast\), so
\[
  |C^\ast|
  =
  16+\sum_{a\in\Fin 8}\varepsilon_i(a).
\]
Here \(16=8\cdot2\): for the fixed coordinate \(i\), the eight fibers
\(C^\ast_{i,a}\), \(a\in\Fin 8\), partition \(C^\ast\), and each contributes
the mandatory baseline of two codewords.
Hence
\[
  \sum_{a\in\Fin 8}\varepsilon_i(a)=6.
\]
Recall that \(\varepsilon_i(a)=|C^\ast_{i,a}|-2\).  Thus
\(\varepsilon_i(a)=0\) means that the fiber \(C^\ast_{i,a}\) has exactly the
mandatory two codewords, while \(\varepsilon_i(a)>0\) means that this fiber
contains additional codewords above that baseline.
Suppose, for contradiction, that there is a coordinate \(i_0\in\Fin 4\) for
which at most five values \(\varepsilon_{i_0}(a)\), \(a\in\Fin 8\), are
positive.  By permuting the coordinate names, assume \(i_0=0\).
Independently, there exists \(b\in\Fin 8\) with \(\varepsilon_1(b)=0\):
the eight nonnegative integers \(\varepsilon_1(x)\), \(x\in\Fin 8\), have
sum six, so at least one of them is zero.  Choose such a \(b\).  Then
\(C^\ast_{1,b}\) contains exactly two codewords.
The only present pairs of the form \((x,b)\) in coordinates \(0,1\) must come
from these two codewords; each such codeword realizes one value of \(x\).
Hence at least six rows \(a\in\Fin 8\) satisfy \((a,b)\in E_{01}\).  Since at
most five coordinate-\(0\) symbols have positive excess, there is an
\(a\in\Fin 8\) with
\[
  (a,b)\in E_{01}
  \qquad\text{and}\qquad
  \varepsilon_0(a)=0.
\]
Equivalently, \(b\in N_{01}(a)\).
We now continue with the same intersection argument as in
Proposition~\ref{prop:no-five-excess}, but with total excess six instead of
five.  Since \(\varepsilon_0(a)=0\) and \(\varepsilon_1(b)=0\), the row
bound~\eqref{eq:missing-neighbor-row-bound} gives
\[
  |N_{02}(a)|,\ |N_{03}(a)|,\ |N_{12}(b)|,\ |N_{13}(b)| \ge 6 .
\]
Thus \(N_{02}(a)\) and \(N_{12}(b)\), as subsets of \(\Fin 8\), intersect in
at least \(6+6-8=4\) points.  Since
\(\sum_x\varepsilon_2(x)=6\), not all four of those points can have
\(\varepsilon_2\)-value at least \(2\); otherwise the sum would be at least
\(8\).  Hence one can choose
\[
  c\in N_{02}(a)\cap N_{12}(b)
  \qquad\text{with}\qquad
  \varepsilon_2(c)\le1.
\]
Then the row bound gives \(|N_{23}(c)|\ge6-\varepsilon_2(c)\ge5\).
Also \(N_{03}(a)\) and \(N_{13}(b)\), as subsets of \(\Fin 8\), intersect in
at least \(6+6-8=4\) points.  Applying the same subset-counting estimate to
this intersection and \(N_{23}(c)\) gives
\[
  |N_{03}(a)\cap N_{13}(b)\cap N_{23}(c)|\ge 4+5-8=1.
\]
Thus one can choose
\[
  d\in N_{03}(a)\cap N_{13}(b)\cap N_{23}(c).
\]
Thus all six coordinate pairs among \(a,b,c,d\) are missing, contradicting
Lemma~\ref{lem:missing-clique}.
Thus exactly six
\(\varepsilon_i(a)\) are positive.  Since these six positive integers have sum
six, they are all equal to one, and the remaining two are zero.

Relabel the alphabet independently in each coordinate so that the six
symbols with \(\varepsilon_i(a)=1\) are \(0,\ldots,5\), and the two symbols
with \(\varepsilon_i(a)=0\) are \(6,7\).  After this relabeling, every missing
graph \(E_{ij}\) satisfies the following two-sided degree bounds: rows and
columns indexed by \(0,\ldots,5\) have at least five entries belonging to
\(E_{ij}\), and rows and columns indexed by \(6,7\) have at least six entries
belonging to \(E_{ij}\).  The
no-clique condition of Lemma~\ref{lem:missing-clique} is invariant under the
same relabeling.

The remaining step is a finite classification.  It is not a hand enumeration
printed in this paper; the complete check is supplied by the two
\Lean-checked LRAT refutations described below.  The finite problem has one
Boolean variable for each ordered pair in each of the six graphs \(E_{ij}\),
hence \(6\cdot8\cdot8=384\) graph variables.  The row and column lower bounds
are encoded as cardinality clauses.  The no-clique condition is encoded by one
clause for each \(a,b,c,d\in\Fin 8\): at least one of the six candidate missing
pairs must be absent.

The proof-relevant finite classification is split into two certified steps.  In
both steps, a \emph{canonical balanced lower profile} means a six-graph profile
after the relabeling above: symbols \(0,\ldots,5\) are the six positive-excess
symbols, symbols \(6,7\) are the two zero-excess symbols, the row and column
lower bounds are those just derived, and the no-clique condition of
Lemma~\ref{lem:missing-clique} holds.  The lower profile does not assume exact
row or column degrees.  In the CNF files, the \emph{lower profile clauses}
are the clauses encoding these row and column lower bounds for the Boolean
variables representing membership \((a,b)\in E_{ij}\).

\paragraph{Finite certificate A: \(E_{01}\) is forced into \(3+3+2\)
cross-block form.}
There is no canonical balanced lower profile in which \(E_{01}\) fails to be a
\(3+3+2\) cross-block graph with low block \(\{6,7\}\) on both sides.
Equivalently, every such profile has two three-symbol high blocks and the
two-symbol low block \(\{6,7\}\) in coordinate \(0\), and similarly in
coordinate \(1\), such that
\[
  E_{01}=\bigcup_{r\ne s} B_0^{(r)}\times B_1^{(s)}.
\]
The CNF file
\leanfilelabel{first lower classifier CNF}
  {data/K_8_4_2/lrat/profile_block_lower_e01_not_block.cnf}
encodes these lower profile clauses, the no-clique clauses, and the negation
of this conclusion by requiring \(E_{01}\) to differ from every candidate
\(3+3+2\) block form.  Its LRAT certificate derives the empty clause, so this
negation is impossible.

\begin{center}
  \captionsetup{type=figure,hypcap=false,font=footnotesize,skip=2pt}
  \CertificateAForcedEdgeDiagram
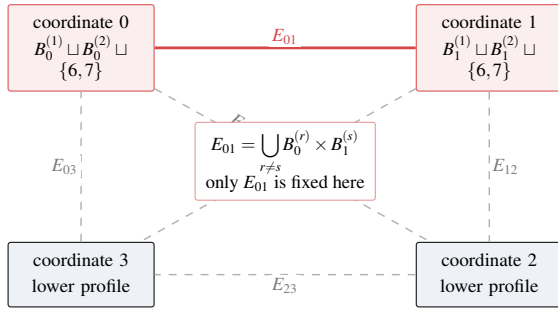
  \caption{Certificate A fixes the \(3+3+2\) cross-block form of \(E_{01}\);
  the other five graphs still only satisfy the lower-profile constraints.}
  \label{fig:certificate-a}
\end{center}

After certificate A, relabel the high symbols in coordinates \(0\) and \(1\) so
that the \(E_{01}\) block form is the fixed canonical one.  This relabeling
preserves the lower profile constraints and the no-clique condition.

\paragraph{Finite certificate B: the block form is global.}
With \(E_{01}\) fixed to this canonical \(3+3+2\) block graph, there is no
canonical balanced lower profile without a common \(3+3+2\) block system for
all six graphs.  Equivalently, every such profile admits block decompositions
for coordinates \(0,1,2,3\) for which each \(E_{ij}\) is exactly the
corresponding cross-block product.
The CNF file
\leanfilelabel{second lower classifier CNF}
  {data/K_8_4_2/lrat/profile_block_lower_not_global.cnf}
keeps the lower profile and no-clique clauses, adds unit clauses forcing the
canonical \(E_{01}\), and encodes the negation of every candidate common block
system by requiring that each candidate fail in at least one graph edge.  Its
LRAT certificate also derives the empty clause.  Hence a common block system
exists.  Undoing the temporary canonical relabeling gives the partitions
\(A_i^{(1)},A_i^{(2)},A_i^{(3)}\) in the statement, and the missing edges are
precisely the cross-block products
\[
  E_{ij}=\bigcup_{r\ne s} A_i^{(r)}\times A_j^{(s)}.
\]

\begin{center}
  \captionsetup{type=figure,hypcap=false,font=footnotesize,skip=2pt}
  \CommonBlockSystemDiagram
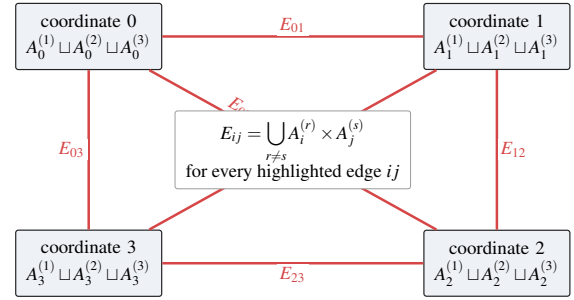
  \caption{Certificate B promotes the forced \(E_{01}\) block form to a common
  \(3+3+2\) block system for all six graphs \(E_{ij}\).}
  \label{fig:common-block-system}
\end{center}

The two lower CNF/LRAT pairs above are the proof-relevant classifier data for
the theorem \(K_8(4,2)\ge23\).  Exact-degree variants are stored in the artifact
as additional checks, but the lower-bound proof uses only the lower profile
constraints stated here.

Thus the printed argument reduces the classifier to two finite unsatisfiability
claims; the actual proof of those claims is the LRAT proof replay checked by
\Lean{}.

\LeanRefs{\leanmaindecl{octonaryFourRadiusTwo_profile_graph_of_card_eq_twenty_two}{4856};
\leanmaindecl{balanced_common_block_system_direct_of_lower_lrat}{4159};
\leanmaindecl{OctonaryFourTwoProfileGraph.E01HasCanonicalLowThreeThreeTwoBlockForm_of_lower_lrat}{2883};
\leanmaindecl{OctonaryFourTwoProfileGraph.HasCommonThreeThreeTwoCrossBlocks_of_lower_lrat}{3227};
\leanlratdecl{profileBlockLowerE01NotBlockProof}{293};
\leanlratdecl{profileBlockLowerNotGlobalProof}{297};
\leanmaindecl{weakRowsCanonicalRelabel4}{5208};
\leanmaindecl{octonaryFourRadiusTwo_profile_graph_relabel_canonicalBalanced}{5252};
\leanmaindecl{no_cover22_of_lower_lrat_classifier}{5295}.}
\end{proof}

\subsection{The common block contradiction}\label{subsec:common-block-contradiction}

\begin{proposition}[Block contradiction]\label{prop:block-contradiction}
No 22-word cover can have the common \(3+3+2\) block system of
Proposition~\ref{prop:profile-classifier}.
\end{proposition}

\begin{proof}
Assume that a 22-word cover \(C^\ast\) has such a common block system.  For
\(r=1,2,3\), define the \(r\)-th block component
\[
  C^{(r)}=\{\,c\in C^\ast : c_i\in A_i^{(r)}
        \text{ for all } i\in\Fin 4\,\}.
\]
Every codeword of \(C^\ast\) belongs to exactly one of these three components.
Indeed, if two coordinates \(i<j\) of a codeword had block indices \(r\ne s\),
then the pair \((c_i,c_j)\) would be a cross-block pair.  By
Proposition~\ref{prop:profile-classifier}, this pair lies in \(E_{ij}\).  But
\(E_{ij}\) consists of pairs which occur in no codeword of \(C^\ast\), a
contradiction.  Hence all four coordinates of \(c\) have the same block index.

Next fix a component \(C^{(r)}\).  Since the only missing pairs are cross-block
pairs, every same-block pair occurs.  More explicitly, if
\((u,v)\in A_i^{(r)}\times A_j^{(r)}\), then this pair is not in \(E_{ij}\), so
some codeword \(e\in C^\ast\) has \(e_i=u\) and \(e_j=v\).  By the previous
paragraph, all coordinates of \(e\) have one common block index; because two of
them lie in the \(r\)-th blocks, this index is \(r\), and hence
\(e\in C^{(r)}\).  Thus for every \(i<j\), the projection of \(C^{(r)}\) to
coordinates \(i,j\) is the full product
\[
  A_i^{(r)}\times A_j^{(r)}.
\]
For \(r=1,2\) these two factors have size three.  Taking only the coordinate
pair \(0,1\), the projection \(C^{(r)}\to A_0^{(r)}\times A_1^{(r)}\) is therefore
surjective onto a set of size \(3\cdot3=9\).  Since one codeword contributes only
one projected pair, \(|C^{(1)}|\ge9\) and \(|C^{(2)}|\ge9\).  The three components are
disjoint subsets of the 22-word cover, so
\[
  |C^{(3)}|\le 22-|C^{(1)}|-|C^{(2)}|\le 4.
\]

It remains to derive a contradiction from the component \(C^{(3)}\) on the
two-symbol blocks.  Choose bijections from \(\Fin 2\) onto the two-symbol sets
\(A_i^{(3)}\), one for each coordinate, and pull \(C^{(3)}\) back to a binary
code \(D\subseteq(\Fin 2)^4\).  Then \(|D|\le4\).  Moreover, the same projection
argument shows that for every coordinate pair \(i<j\), all four binary pairs
occur in the projection of \(D\) to \(i,j\).  In other words, \(D\) would be a
binary strength-two array of length four with at most four rows, where
``strength two'' means that every projection to two coordinate columns contains
all four binary pairs.

No such binary array exists.  Since the first two columns must realize all four
binary pairs, \(|D|\ge4\), hence \(|D|=4\).  After reordering rows and possibly
complementing the first two columns, the first two columns may be written as
\[
  00,\quad 01,\quad 10,\quad 11.
\]
Write a third column as \(y=(y_1,y_2,y_3,y_4)\).  To realize both pairs with the
first column, one must have \(y_1\ne y_2\) and \(y_3\ne y_4\).  To realize both
pairs with the second column, one must have \(y_1\ne y_3\) and
\(y_2\ne y_4\).  These four inequalities force
\[
  y=(0,1,1,0)
  \qquad\text{or}\qquad
  y=(1,0,0,1).
\]
Thus any third or fourth column that realizes all four pairs with each of the
first two columns must be one of these two complementary columns.  The third and
fourth columns are therefore either equal or complementary.  In the first case
their joint projection contains only \(00\) and \(11\); in the second case it
contains only \(01\) and \(10\).  In neither case are all four binary pairs
present, a contradiction.

\LeanRefs{\leanmaindecl{no_code_with_common_three_three_two_missing_cross_block_system}{4422};
\leanmaindecl{no_cover22_of_extracted_common_cross_blocks}{4984};
\leanmaindecl{no_binary_four_column_strength_two_at_most_four_rows}{882};
\leanmaindecl{no_binary_four_column_strength_two_four_rows}{859}.}
\end{proof}

\begin{theorem}[Lower bound]\label{thm:k842-lower}
\[
  K_8(4,2)\ge 23.
\]
\end{theorem}

\begin{proof}[Proof of Theorem~\ref{thm:k842-lower}]
Let \(C^\ast\) be a hypothetical radius-two cover with \(|C^\ast|\le22\).  By
Proposition~\ref{prop:no-five-excess}, \(|C^\ast|\) cannot be at most 21.
Hence \(|C^\ast|=22\).  Proposition~\ref{prop:profile-classifier} then gives a
common \(3+3+2\) block system for the missing-pair profile of \(C^\ast\), but
Proposition~\ref{prop:block-contradiction} shows that no 22-word cover can have
such a profile.  This contradiction proves \(K_8(4,2)\ge23\).

\LeanRefs{\leanmaindecl{octonaryFourRadiusTwoLowerTwentyThree}{5497}.}
\end{proof}

\section{Formalization Contribution}

\begin{table*}[t]
\centering
\scriptsize
\setlength{\tabcolsep}{3pt}
\begin{tabular}{@{}p{0.22\textwidth}p{0.30\textwidth}p{0.20\textwidth}p{0.24\textwidth}@{}}
\toprule
Paper claim & Lean declaration or data & Artifact link & Trust boundary \\
\midrule
Explicit 23-word upper bound &
\leanmain{explicit code}{36},
\leanmain{upper validity}{83} &
\leanmain{main source}{1} &
Finite covering check in \Lean{}. \\

Helper lower bound \(K_7(3,1)\ge22\) &
\leansparse{\(K_7(3,1)\ge22\)}{400} &
\leansparse{sparse-slicer source}{1} &
Handwritten \Lean{} plus finite leaves. \\

Lower bound excluding \(\le21\) words &
\leanmain{21-word exclusion}{5455} &
\leanmain{main source}{1} &
Structural graph argument in \Lean{}. \\

LRAT profile classifiers &
\leanlrat{\(E_{01}\) LRAT proof}{293},
\leanlrat{global LRAT proof}{297} &
\leanlrat{LRAT source}{1},
\leandirlabel{LRAT data}{data/K_8_4_2/lrat/} &
CNF/LRAT parsed and checked in \Lean{}. \\

Exact lower bound \(K_8(4,2)\ge23\) &
\leanmain{exact lower theorem}{5497} &
\leanmain{main source}{1} &
Structural proof plus LRAT classifiers. \\

Exact value \(K_8(4,2)=23\) &
\leancovernum{\texttt{QaryKSpec}}{40},
\leancovernum{\texttt{KSpec.ofUpperLower}}{89} &
\leancovernum{covering-number source}{1} &
Exactness obtained by recombining matching upper and lower certificates. \\
\bottomrule
\end{tabular}
\caption{Compact audit map for the \(K_8(4,2)=23\) certificate.}
\label{tab:k842-audit}
\end{table*}

\begin{table*}[t]
\centering
\footnotesize
\begin{tabular}{@{}llrrp{0.28\textwidth}@{}}
\toprule
Proof-path item & Mode & Wall time & Max RSS & Comment \\
\midrule
\leansource{\texttt{SparseSlicer.lean}}{CoveringCodes/Database/Sources/SparseSlicer.lean}{1} &
kernel & 5.58 s & 3.65 GiB &
Helper lower bound \(K_7(3,1)\ge22\). \\
\leansource{\texttt{OctonaryFourTwoBlockLRAT.lean}}{CoveringCodes/Database/Sources/OctonaryFourTwoBlockLRAT.lean}{1} &
kernel & 17.32 min & 13.44 GiB &
CNF/LRAT parsing and reflective LRAT proof construction. \\
\leansource{\texttt{OctonaryFourTwo.lean}}{CoveringCodes/Database/Sources/OctonaryFourTwo.lean}{1} &
kernel & 13.67 min & 143.40 GiB &
Main \(K_8(4,2)\) proof path with finite leaves replayed by kernel reduction. \\
\leansource{\texttt{OctonaryFourTwo.lean}}{CoveringCodes/Database/Sources/OctonaryFourTwo.lean}{1} &
native & 2.92 min & 5.54 GiB &
Same main proof path with switchable finite leaves closed by native evaluation. \\
\bottomrule
\end{tabular}
\caption{Focused resource measurements for the \(K_8(4,2)=23\) proof path.
Times are wall-clock measurements and memory is maximum resident set size for
whole-file checks.}
\label{tab:k842-resources}
\end{table*}

At the theorem-prover boundary, the result uses the same covering-number
certificate interface as the general database.  The \Lean{} library uses the
linked \leancovernum{upper-bound predicate}{54} and
\leancovernum{lower-bound predicate}{68}.  Exact statements are expressed by
the linked \leancovernum{\texttt{QaryKSpec}}{40} predicate.  The present
artifact does not introduce a separate primitive exact certificate; its citable
exactness step is the library theorem
\leancovernum{\texttt{KSpec.ofUpperLower}}{89}, applied to the matching upper
and lower certificates at \(q=8\), \(n=4\), \(r=2\), and \(k=23\).

For this paper, the upper-bound source is registered under the primitive database
label
\begin{center}
  \leanmaindecl{lean_octonary_four_two_explicit_upper}{73}
\end{center}
It is backed by the declaration
\leanmain{upper-bound validity theorem}{83}.
The listed code is
\leanmain{explicit code data}{36}, packaged by
\leanmain{explicit upper wrapper}{75}.  The wrapper is an instance of the linked
\leanexplicit{explicit upper-bound database structure}{7}.  This structure
stores the field \texttt{card\_le}; thus the formal upper certificate proves the
upper bound \(K_8(4,2)\le23\), not a separate cardinality-equality theorem for
the displayed list.
The lower-bound source is registered under
\begin{center}
  \leanmaindecl{lean_octonary_four_two_structural_lower}{5529}
\end{center}
It is backed by
\leanmain{lower-bound validity theorem}{5534}.
At the target parameters these specialize respectively to
\begin{center}
  \leancovernum{\texttt{QaryKUpper}}{54} \(\texttt{8 4 2 23}\),\\
  \leancovernum{\texttt{QaryKLower}}{68} \(\texttt{8 4 2 23}\).
\end{center}
Together with the library lemma
\leancovernum{upper/lower recombination lemma}{89}, these are
the formal certificate of \(K_8(4,2)=23\).

The proof-local \Lean{} references above identify the individual mathematical
ingredients.  This section summarizes the formal architecture that connects
them.  The upper bound is a conventional explicit-code certificate.  The lower
bound is different:
it starts with an arbitrary finite set \(C\) satisfying the covering predicate
and derives the coordinate-fiber inequalities, the above-baseline profiles, and
the six missing-pair graphs inside the theorem prover.

The main bridge introduced for this proof is the linked
\leanmain{profile graph type}{989}.  It records the six missing-pair graphs
together with their degree information and the no-clique constraint.  A
hypothetical 22-word cover is converted into such a profile by the linked
\leanmain{profile-graph extraction theorem}{4856}.
This is the formal interface between the mathematical cover and the finite
profile-classification problem.

The lower-bound proof then separates three kinds of evidence.  First, ordinary
\Lean{} arguments extract the balanced profile constraints from a hypothetical
cover.  Second, reflective LRAT replay proves the finite classifier statements
for those profiles.  Third, ordinary \Lean{} combinatorics turns the resulting
common \(3+3+2\) block system into a contradiction.  The SAT solver is therefore
only an untrusted producer of CNF instances and LRAT data; the checked theorem
depends on the stored certificates and on the \Lean{} proofs that connect them
to covering codes.

The proof also records a reusable pattern for computational lower bounds in the
covering-code database.  A search can suggest a small finite obstruction, but
the artifact must still provide a typed extraction theorem, a certificate
checker or reflective proof object, and a final contradiction stated in the
same \leancovernum{lower-bound predicate}{68} interface used
by the rest of the database.  This makes the exact value \(K_8(4,2)=23\) a
test case for attaching solver-produced evidence to a theorem-prover
certificate without adding the solver to the trusted replay path.

\section{Certificate Boundary and Replay}

The trusted boundary is split into four separate stages.

\smallskip\noindent\textbf{Discovery.}\par
Untrusted generation tooling produces the CNF instances, and an untrusted SAT
proof-producing solver produces the LRAT refutations.  Neither the generator,
the solver run, nor any external SAT checker is part of the trusted proof.

\smallskip\noindent\textbf{Frozen artifact.}\par
The checked artifact consists of the handwritten \Lean{} arguments, the
explicit 23-word code, the stored CNF formulas, and the stored LRAT
refutations.  The two lower-bound profile-classification stages used in the
main proof are represented by
\leanfilelabel{CNF 1}
  {data/K_8_4_2/lrat/profile_block_lower_e01_not_block.cnf}
and
\leanfilelabel{CNF 2}
  {data/K_8_4_2/lrat/profile_block_lower_not_global.cnf}.
Their clause counts are 10880 and 10944, respectively.  Exact-degree variants
are stored as additional audit data; all four CNF/LRAT pairs occupy about
11 MB in the linked
\leandirlabel{LRAT data}{data/K_8_4_2/lrat/}.

\smallskip\noindent\textbf{Replay in \Lean.}\par
During proof replay, \Lean{} parses the stored certificate data and checks the
LRAT derivations internally.  The CNF generator is not trusted: for each stored
lower-bound CNF used in the proof, \Lean{} proves that the parsed CNF is exactly
the corresponding direct \Lean{} formula, and separately proves that any
canonical lower profile violating the claimed block conclusion induces a
valuation satisfying that formula.  The LRAT replay then proves that no such
valuation exists.  The bridge from these finite contradictions to the 22-word
cover obstruction is the linked
\leanmain{22-word classifier contradiction}{5295}.  Table~\ref{tab:k842-audit}
summarizes where the paper claims enter the checked artifact.

\LeanRefs{\leanlratdecl{profileBlockLowerE01NotBlockFormula_eq_direct}{436};
\leanlratdecl{profileBlockLowerNotGlobalFormula_eq_direct}{485};
\leanmaindecl{OctonaryFourTwoProfileGraph.profileBlockValuation_satisfies_lowerE01NotBlockFormulaDirect}{2856};
\leanmaindecl{OctonaryFourTwoProfileGraph.profileBlockValuation_satisfies_lowerNotGlobalBlockFormulaDirect}{3196}.}

\smallskip\noindent\textbf{Performance option.}\par
The finite explicit-code covering checks use the configurable proof mode
described in the formal-foundations
paper~\cite{florath2026formal_foundations_covering_codes}.
The headline proof path is replayable in kernel mode: the kernel-mode row in
Table~\ref{tab:k842-resources} uses ordinary \texttt{decide} and does not rely
on \texttt{native\_decide}.  For development and faster artifact checks, native
mode may close selected expensive finite leaves by \texttt{native\_decide};
this is a performance option and enlarges the trusted base to include the
\Lean{} native compiler and generated code.  The LRAT wrapper itself is not a
native-decide shortcut: its central output is a proof term checked by \Lean{}
from the parsed LRAT data.

Table~\ref{tab:k842-resources} records focused file-level measurements for this
proof path.  The runs follow the same convention as the formal-foundations
paper: Lean was invoked through Lake with one worker and the
reported memory is maximum resident set size for the whole file-level check.
These are isolated file-level proof-path measurements in an already built
checkout, not clean-build timings for the full repository.  The full kernel
replay of the main \(K_8(4,2)\) file is therefore a high-memory check: in this
measurement it required 143.40 GiB of RAM.  Native mode is the practical
lower-memory replay option for the expensive finite leaves.

The main file is the only part of this proof path with a large native-vs-kernel
difference, because it contains expensive finite proof leaves closed by the
switchable proof mode.  The LRAT wrapper is dominated by parsing, proof
construction, and kernel checking of the certificate data, so native mode is
not expected to give the same kind of improvement there.

The measurements were taken on the same server as reported in the
formal-foundations paper: two
Intel Xeon Gold 6438M sockets, 64 physical cores / 128 hardware threads, and
512 GiB RAM.

The focused database query recorded for this result is:
\begin{lstlisting}[language=bash]
$ ./.lake/build/bin/covering_codes 8 4 2
K_8(4,2) ∈ [23, 23]
lower:
  1) lean_octonary_four_two_structural_lower
upper:
  1) lean_octonary_four_two_explicit_upper
(exact: K_8(4,2) = 23)
\end{lstlisting}

\section{Discussion}

The exact value \(K_8(4,2)=23\) is the numerical endpoint of the paper, but the
proof method is the more reusable part.  The lower bound does not reason
directly about all possible 22-word codes.  Instead it passes to complement
data: for each coordinate pair it records which alphabet pairs are missing from
the projection of the code.  In this inverse viewpoint, a word not covered at
radius two appears as a four-partite clique in the missing-pair graphs.  The
covering property therefore becomes a graph-theoretic no-clique condition,
while the fiber bounds become degree constraints on those graphs.

This structural reduction is useful because it separates the proof into a
human-readable mathematical part and a finite classification part.  The
mathematical argument explains why every hypothetical small cover must give one
of a tightly constrained family of missing-pair profiles.  The remaining
classification is not presented as a hand-checkable enumeration.  It is supplied
as CNF instances together with LRAT refutations, and \Lean{} checks that these
certificates derive the required contradictions.

The LRAT workflow changes the role of computation in the proof.  The SAT solver
is used to discover and emit certificate data, but it is not trusted during
proof replay.  The trusted artifact is the \Lean{} development that parses the
stored CNF/LRAT data, checks the refutations, and connects their conclusions
back to the covering-code statement.  Thus a finite subproblem may be too large
for a reader to validate case by case, while still being small enough to become
a theorem-prover-checked proof object.

This suggests a practical pattern for future covering-code lower bounds.  One
can first search for a structural compression of the covering problem, then use
SAT or related tools to solve the resulting finite profile problem, and finally
connect the resulting certificate to the covering-number interface in \Lean{}.
For \(K_8(4,2)\), this pattern closes a single table entry.  Its broader value
is that it shows how structural combinatorics and proof-carrying computational
evidence can be combined without adding the solver itself to the trusted base.

\section{Conclusion}

The exact value \(K_8(4,2)=23\) follows from two machine-checkable certificates:
an explicit radius-two 23-word octonary code, and a \Lean{} lower-bound proof
excluding all 22-word covers.  The lower-bound proof compresses any hypothetical
22-word cover to a six-graph profile problem, uses \Lean-checked LRAT
refutations of stored CNF instances to force a common \(3+3+2\) block system,
and then finishes with a
small combinatorial contradiction inside \Lean{}.

The result also demonstrates the certificate interface of the covering-code
formalization at a larger level of complexity.  Explicit codes, structural lower
bounds, and solver-produced certificates all enter as ordinary proof objects
that can be recombined by the database.  For this parameter, that proof-carrying
workflow turns the historical interval \([22,23]\) into the exact value 23.

\section*{Acknowledgements}
\phantomsection
\addcontentsline{toc}{section}{Acknowledgements}

This work was partly funded by the Federal Ministry of Research, Technology and
Space (BMFTR) in Germany under grant number 16KIS2240 of the SUSTAINET-guardian project.

\section*{Use of AI Tools}
\phantomsection
\addcontentsline{toc}{section}{Use of AI Tools}

This work used AI-based assistants (in particular Codex and Claude
Code) at several stages, and their role is described here explicitly.

After the underlying \(q\)-ary covering-code formalization was complete,
these tools were used to explore which entries of the standard bound
tables might be amenable to a new formal proof on top of that
infrastructure.  This exploration suggested the octonary case
\(K_8(4,2)\) as a candidate.  Unguided search by the tools did not
produce a proof.  The decisive step --- reformulating the covering
condition in terms of the coordinate-pair missing-pair graphs
\(E_{ij}\), so that a hypothetical small cover appears as a forbidden
clique --- was directed by the author.  With that structural framing
fixed, AI tools assisted in working out the supporting lemmas and in
drafting parts of the \Lean{} development and the auxiliary tooling.
AI tools were also used to draft and revise portions of the manuscript
and to assist in several internal review passes.

Formal statements in the artifact are checked by the \Lean{} kernel,
independently of how the corresponding proof terms were produced.  The
author traced the formal argument used for this result, rewrote parts of it
where needed, verified its correspondence with the informal proof presented
here, and reviewed the mathematical statements, cited \Lean{} declarations,
citations, and manuscript text.  The author takes full responsibility for the
content of the paper.

\onecolumn
\phantomsection
\sloppy
\Urlmuskip=0mu plus 1mu\relax
\printbibliography

\end{document}